\documentclass{svproc}
\usepackage{url}
\usepackage{graphicx}

\usepackage{amsmath}
\usepackage{comment}

\begin{document}
\mainmatter              % start of a contribution
\title{Negative correlations in Ising models of credit risk}
\titlerunning{Negative correlations in credit risk}  % abbreviated title (for running head)
\author{Chiara Emonti\inst{1} \and Roberto Fontana\inst{2}
}
\authorrunning{Chiara Emonti and Roberto Fontana} % abbreviated author list (for running head)
\tocauthor{Chiara Emonti, Roberto Fontana}
\institute{Politecnico di Torino, Corso Duca degli Abruzzi 24, Torino, Italy\\
\email{s316686@studenti.polito.it}
\and
Politecnico di Torino, Corso Duca degli Abruzzi 24, Torino, Italy\\
\email{roberto.fontana@polito.it}}

\maketitle              % typeset the title of the contribution

\begin{abstract}
We analyze a subclass of Ising models in the context of credit risk, focusing on Dandelion models when the correlations $\rho$ between the central node and each non-central node are negative. We establish the possible range of values for $\rho$ and derive an explicit formula linking the correlation between any pair of non-central nodes to $\rho$. The paper concludes with a simulation study.
% We would like to encourage you to list your keywords within
% the abstract section using the \keywords{...} command.
\end{abstract}
\section{Credit Risk}
Credit risk arises in the context of financing and is one of the most important factors in determining the price of financial securities. It is also considered a key element in the daily investment decisions of investors, whether banks or private individuals. As highlighted in \cite{porretta2021integrated}, credit risk refers to the possibility that a counter-party fails to fully meet its obligations, both in terms of interest and principal. In such cases, the counter-party is said to be in default. A higher credit risk leads to a higher required interest rate.  
We consider a portfolio of $N$ credits. For each credit $i$, $i=1,\ldots,N$, three key quantities are defined: Probability of Default (PD), Loss Given Default (LGD), and Exposure at Default (EAD). By combining these quantities, we obtain the random variable $L_i'$, which represents the loss associated with the $i$-th credit, $L_i' = EAD_i \cdot LGD_i \cdot L_i$, where $L_i \in \{0,1\}$ is a Bernoulli-distributed variable with parameter $PD_i$, indicating whether a default occurs for the $i$-th credit. The total credit portfolio loss $L'$ is then given by $L'=L_1'+\ldots+L_N'$.
In this work, we assume that both $EAD_i$ and $LGD_i$ are constant. Under this assumption, the total loss $L'$, apart from a negligible constant term, simplifies to  
\begin{equation}
L=\sum_{i=1}^{N} L_i
\end{equation}

%\subsection{Existing credit portfolio models and contagion}
%
As stated in \cite{molins2016model}, a credit portfolio model is a theoretical construct that outputs the probability distribution of losses for a given credit portfolio. The continuous evolution of financial methodologies has exposed the limitations of traditional models in accurately assessing risks. In particular, none of these models was able to adequately capture extreme risks during the US subprime crisis. This shortcoming arises because traditional portfolio models fail to account for contagion effects.  
In this work, we study a class of models that incorporate contagion effects, \cite{molins2016model}. Financial contagion is defined as "a significant increase in co-movements of prices and quantities across markets, conditional on a crisis occurring in one market or group of markets", \cite{pericoli2003primer}. Readers interested in the subject can  refer to \cite{molins2016model} and \cite{fontana2021model} and the references therein.

\section{Jungle Model}

In the literature, the Jungle model is known as the Ising model, originally derived from statistical mechanics. This model is used to model credit contagion by generating double-peak probability distributions for credit loss and endogenously producing quasi-phase transitions, which are examples of systemic credit crises that arise suddenly without an apparent cause, \cite{molins2016model}.  

The Jungle model provides the optimal probability distribution for modeling losses in a general credit portfolio, based on two assumptions: (i) the principle of Maximum Entropy (MaxEnt) and (ii) the premise that all empirical information of a given credit portfolio can be summarized by the probability of default and the correlations among defaults of its components, \cite{molins2016model}. This approach demonstrates how probability distributions can be determined from partial information. By utilizing the available data and ensuring that the unknown distribution maximizes Shannon entropy, the probability distribution that best describes the given scenario can be derived.  
Following \cite{molins2016model}, the MaxEnt principle seeks to find the probability distribution $P(x)$, defined over a finite state space $\Omega$, that maximizes entropy while satisfying $m$ constraints, given by:
\begin{equation}
\sum_{x \in \Omega} P(x) f_k(x) = F_k 
\end{equation}
where $f_k(x)$ is a function defined on $\Omega$, and $F_k$ is a given constant, for $k=1,\ldots,m$, with $m < \#\Omega$. In addition, the standard constraint $\sum_{x \in \Omega} P(x) = 1$ must be satisfied. The solution to this problem is given by
\begin{equation}
P(x)=\frac{1}{Z(\lambda_1,\ldots,\lambda_m)} \exp \left(-\sum_{i=1}^{m}\lambda_i f_i(x)\right)
%P(x)=\frac{1}{Z(\lambda_1,\ldots,\lambda_m)}\exp⁡\left(-\sum_{i=1}^{m}\lambda_i f_i(x)\right)
\end{equation} 
where $Z$ is the partition function $Z(\lambda_1,\ldots,\lambda_m)=\sum_{x \in \Omega} \exp(\sum_{i=1}^{m} \lambda_i f_i (x))$. The parameters $\lambda_1,\ldots,\lambda_m$ are determined by solving the equations $F_k=-\frac{\partial \log Z}{\partial\lambda_k}$, $k=1,\ldots,m$.

\subsection{The Dandelion Model}
The Dandelion model is a member of the Jungle model class and represents a situation in which a central element of the credit portfolio is linked to each of the other non-central nodes. When the correlation between the default indicator of the central node and the default indicator of each non-central node is positive, the model mimics the relationship between a bank and its borrowers, or between a central bank and the rest of the economy.

The probability function of the Dandelion model is given by:
\begin{equation} \label{eq:p0N}
P(l_0, l_1, \ldots, l_N) = \frac{1}{Z} \exp \left( \alpha_0 l_0 + \alpha \sum_{i=1}^{N} l_i + \beta \sum_{i=1}^{N} l_0 l_i \right)
\end{equation}
where the parameters $\alpha_0$, $\alpha$, and $\beta$ are determined to satisfy the requirements $E[L_0] = p_0$, $E[L_i] = p_d$, and $E[L_0 L_i] = q$ for $i = 1, \ldots, N$, and $Z=Z(\alpha_0,\alpha,\beta)$ is the normalizing function. In the case where all default probabilities are equal, $p_0 = p_d = p$, using \cite{molins2016model}, we obtain:
\begin{equation}\label{eq:param}
\begin{aligned}
\alpha=\log \left(\frac{p-q}{1-2p+q}\right) , \;
\alpha_0=(N-1)\log\left(\frac{1-p}{p}\right)+N\alpha, \;
\beta=\log\left(\frac{q}{p-q}\right)-\alpha , \\
Z=(1+\exp(\alpha))^N+\exp (\alpha_0)(1+\exp(\alpha+\beta))^N
\end{aligned}
\end{equation}

From Eq.(\ref{eq:p0N}) we can derive the expression of the probability mass function of the loss $L=L_1+\ldots+L_N$:
\begin{equation}
  P(L=l)=\frac{1}{Z} \binom{N}{l} \left( \exp(\alpha l)+\exp(\alpha_0+l(\alpha+\beta)) \right), \; l=0,\ldots, N
\end{equation}
We recall the relation between the second order moment $q=E[L_0L_i]$ and the correlation $\rho=\rho(L_0,L_i)$: $\rho=\frac{q-p^2}{p(1-p)}, q=\rho p (1-p)+p^2$, $i=1,\ldots,N$.
\begin{comment}
\begin{equation} \label{eq:rho} 
\rho=\frac{q-p^2}{p(1-p)}   \; \Leftrightarrow \; q=\rho p (1-p)+p^2
\end{equation}    
\end{comment}

\begin{comment}
/*** rivedere
When the correlation between the credits is very low, close to 0, the generated distribution is very similar to a binomial, that is the distribution of the losses in case of independence between the credits; as the default correlation increases, a second peak of losses can be observed, which moves progressively to the right as the severity of possible losses increases.    
\end{comment}

\begin{comment}
\begin{figure}
    \centering
    \includegraphics[width=0.45\linewidth]{image1.png}
    \caption{Probability distribution for a credit portfolio with N=100 and PD=0.2 taking into account the contagion effect for increasing correlation}
    \label{fig:enter-label}
\end{figure}    
\end{comment}
%
\subsection{Negative correlations in Dandelion model}
In this work, we explore the case where the correlation $\rho$ between the default indicator of the borrower $L_0$ and the default indicator of the $i$-th creditor $L_i$ is negative, $i=1,\ldots,N$. The paper \cite{molins2016model} considers the Dandelion model only in the case of $\rho$ greater than zero. 
The status of non-central nodes may be negatively correlated with that of the central node. For example, consider a case where the central node represents an oil-trading company aiming to expand by acquiring $N$ companies in the renewable energy sector. In this scenario, it seems reasonable that the statuses of the central and non-central nodes would be negatively correlated.

\begin{comment}
if one security is linked to the evolution of the price of coal and another is linked to the use of alternative energy sources, it is plausible that if the first security goes into default, the second not and will probably increase in value.
\end{comment}

We know that the second-order moment $q$ can take values in the interval $[0, p]$, i.e., $0 \leq q \leq p$. It follows that $\rho$ is bounded:
\begin{equation} \label{eq:genrobound}
    -\frac{p}{1-p} \leq \rho \leq 1
\end{equation}

In addition, the parameters $\alpha$, $\alpha_0$, and $\beta$ are defined using the logarithm. Proposition \ref{prop:rangelog} provides the range of admissible values for the correlation $\rho$.

\begin{proposition} \label{prop:rangelog}
Given $0<p<1$ and $0<q<p$, the correlation $\rho$ must satisfy the constraints
\[
\rho > \max( -\frac{p}{1-p}, -\frac{1-p}{p}), \quad \rho < 1
\]
\end{proposition}
\begin{proof}
Let us consider $\alpha=\log(\frac{p-q}{1-2p+q})$. Since $p-q>0$, it must be that $1-2p+q>0$. We obtain
\[
1-2p+q>0 \Leftrightarrow q > 2p-1 \Leftrightarrow \rho p(1-p)+p^2>2p-1 \Leftrightarrow \rho > - \frac{1-p}{p}
\]
It is easy to verify that $\alpha_0$ and $\beta$ do not impose additional constraints. Therefore, from Eq.(\ref{eq:genrobound}), the thesis follows.
\end{proof}

We know that $E[L_0]=E[L_i]=p$ and $E[L_0 L_i]=q$, $i=1,\ldots,N$. Now we study the second-order moment $E[L_iL_j]$ and the correlation $\rho_{i,j}$ between two non-central creditors $i$ and $j$, with $i\neq j$ and $i,j \in {1,\ldots, N}$.

\begin{proposition} \label{prop:secmom}
Let $\rho$ ($\rho_{i,j}$) denote the correlation between $L_0$ and $L_k$ ($L_i$ and $L_j$, respectively) where, $i,j,k \in \{1,\ldots,N\}, i \neq j$. It follows that
\[
\rho_{i,j}=\rho^2
\]
\end{proposition}
\begin{comment}
    \[
    E[L_iL_j]=\frac{1}{Z} \left[ e^{2\alpha}(1+e^\alpha)^{N-2}
    +e^{[\alpha_0+2(\alpha+\beta)]}(1+e^{\alpha+\beta})^{N-2}\right]
    \]   
\end{comment}
 \begin{proof}
  From Eq.(\ref{eq:p0N}) we obtain
  \[
  P'(l_1,\ldots,l_N)=\sum_{l_0=0}^{1} P(l_0,l_1,l_2,...,l_N)=\frac{1}{Z}\left[\exp\left(\alpha \sum_{i=1}^Nl_i\right) + \exp\left(\alpha_0+(\alpha+\beta)\sum_{i=1}^Nl_i\right)\right]
  \]
  Without loss of generality, we consider $i=1$ and $j=2$. We obtain
  \begin{eqnarray*}
    E[L_1 L_2]=P'(L_1=1,L_2=1)=\sum_{(l_3,\ldots,l_N)\in V_{N-2}}P'(1,1,l_3,\ldots,l_N)= \\
    \sum_{(l_3,\ldots,l_N) \in V_{N-2}}\frac{1}{Z}\left[\exp\left(2\alpha+\alpha \sum_{i=3}^Nl_i\right) + \exp\left(\alpha_0+2(\alpha+\beta)+(\alpha+\beta)\sum_{i=3
    }^Nl_i\right)\right]      
  \end{eqnarray*}
  where $V_{N-2}=\{0,1\}^{(N-2)}$. 
  Using the parameter expressions from Eq.~(\ref{eq:param}), the first addendum of the above expression becomes
 \begin{eqnarray*}
  \sum_{(l_3,\ldots,l_N)}\frac{1}{Z}\left[\exp\left(2\alpha+\alpha \sum_{i=3}^Nl_i\right)\right]
  =\frac{1}{Z}e^{2\alpha} \sum_{(l_3,\ldots,l_N) \in V_{N-2}} e^{\alpha (l_3+\ldots+l_N)}= \\
  =\frac{1}{Z}e^{2\alpha} \sum_{l_3=0}^1 e^{\alpha l_3} \ldots \sum_{l_N=0}^1 e^{\alpha l_N} 
  =\frac{1}{Z}e^{2\alpha} (1+e^\alpha)^{N-2}=\frac{(p-q)^2}{1-p}
 \end{eqnarray*}
With similar computations for the second addendum, we obtain
\[
\sum_{(l_3,\ldots,l_N) \in V_{N-2}}\frac{1}{Z}\left[ \exp\left(\alpha_0+2(\alpha+\beta)+(\alpha+\beta)\sum_{i=3
    }^Nl_i\right)\right]=\ldots=\frac{q^2}{p} 
\]
%\frac{1}{Z} \left[ e^{[\alpha_0+2(\alpha+\beta)]}(1+e^{\alpha+\beta})^{N-2}\right]=
It follows
\[
\rho_{12}=\frac{E[L_1L_2]-p^2}{p(1-p)}=\frac{(p^2-q)^2}{p^2(1-p)^2}
\]
Since $\rho=\frac{q-p^2}{p(1-p)}$ the proof is completed.
\end{proof}
It follows that the correlation between two non-central nodes is always positive and smaller than $|\rho|$.  

\begin{comment}
To ensure that the model is correct, it was necessary to check the domain of $\rho$ such that the logarithm arguments in the expressions of $\alpha$, $\alpha_0$ and $\beta$ are positive. We report the equations extracted by \cite{molins2016model} considering p = p0 (p is the probability of default of the central node):

\begin{equation}
\alpha_0=(N-1)  log⁡((1-p)/p)+N log((p-q)/(1-2p+q))
\end{equation}

\begin{equation}
\alpha=log⁡((p-q)/(1-2p+q))
\end{equation}

\begin{equation}
\beta=log⁡(q/(p-q)  (1-2p+q)/(p-q))
\end{equation}

As we can see, the following requirements must be met:
\begin{equation}
(1-p)/p>0
\end{equation}

\begin{equation}
(p-q)/(1-2p+q)>0
\end{equation}

\begin{equation}
q/(p-q)   (1-2p+q)/(p-q)  =  (q(1-2p+q))/(p-q)^2 >0
\end{equation}

All results combined show that the Dandelion model with negative correlations can be solved if the following steps are taken:
\begin{equation}
(ρ>p/(p-1) se p<0,5  ρ> (p-1)/p    se p>0,5)
\end{equation}
\end{comment}

\section{Simulation}

\begin{comment}
After integrating these results into the code for the Dandelion model, the following graphs are obtained.
\begin{figure}
    \centering
    \includegraphics[width=0.5\linewidth]{image2.png}
    \caption{Probability distribution with $N=100$ and $p_0=p_d=p=0.4$ with $-1<\rho<1$?} 
    \label{fig:enter-label}
\end{figure}
\end{comment}
We consider the case with $p = 0.4$ and $N = 100$. The probability of default is high and somewhat unrealistic, but as in \cite{molins2016model}, it is a good example for illustrative purposes. From Proposition \ref{prop:rangelog}, we find that the range of correlations covered by the model is $-\frac{2}{3} < \rho < 1$. From Proposition \ref{prop:secmom}, we obtain that in the case of negative (positive) correlations, $\rho < 0$ ($\rho > 0$), the correlation $\rho_{ij}$ will lie in the range $(0,\frac{4}{9})$ (or $(0, 1)$), respectively. 
The loss distribution presents two peaks in the case of both negative and positive correlation $\rho$. Figure \ref{fig:2pmf} shows two loss distributions: one with positive correlation $\rho = \rho_0$ and the other with negative correlation $\rho = -\rho_0$ with $\rho_0=0.26$. The distribution for positive $\rho$ has the first (highest) peak at low losses and the second (lowest) peak at high losses. In contrast, for negative $\rho$, the opposite is true. Therefore, it appears that when the correlations are negative, a serious portfolio loss has a high probability of occurring.

\begin{figure} 
\begin{minipage}[c]{0.4\linewidth}
\includegraphics[width=\linewidth]{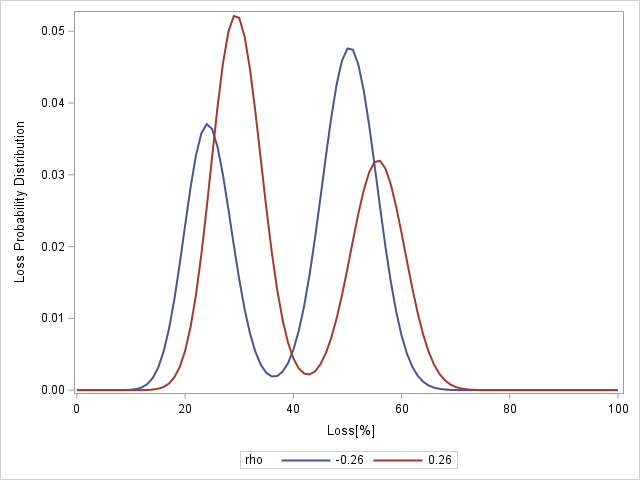}
\caption{Probability Loss Distribution for $\rho=\rho_0$ and $\rho=-\rho_0$ with $\rho_0=0.26$}
\label{fig:2pmf}
\end{minipage}
\hfill
\begin{minipage}[c]{0.4\linewidth}
\includegraphics[width=\linewidth]{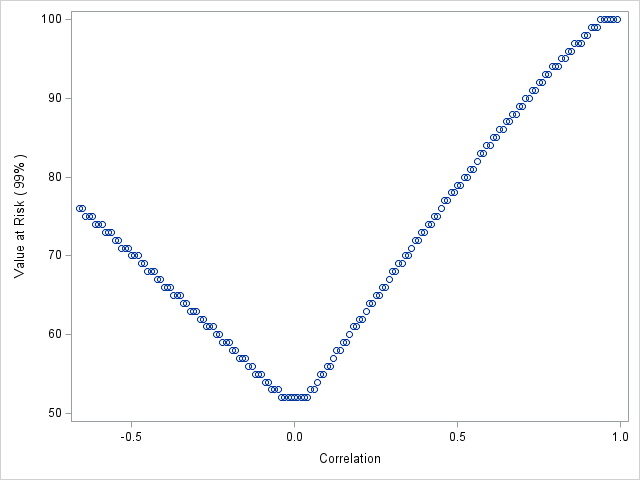}
\caption{Value at Risk (99\%) \emph{vs} Correlation}
\label{fig:var}
\end{minipage}%
\end{figure}

In Figure \ref{fig:var}, the relationship between the Value at Risk (99\%) and the correlation is shown. The graph is not perfectly symmetrical, even though it appears to be \emph{mirrored} around zero.
To further study this phenomenon, we considered the modes of the loss distributions. The graph in Figure \ref{fig:2figsmode} shows the relationship between the mode of the loss distribution and the default correlation $\rho$. It is evident that a discontinuity exists: there is a \emph{jump} at the value $\rho_\star \approx -0.4$. 
Upon further analysis of the mode's trend, we observe that initially, when $\rho$ is close to the lower bound ($\rho \approx - \frac{2}{3}$), the mode is almost zero. As $\rho$ increases, the mode starts to increase, and when $\rho$ is close to $\rho_\star$ it \emph{jumps} to its maximum value, around $60$. After reaching this maximum, it decreases almost linearly to $0$. 
This sudden change in the mode can be interpreted as a quasi-phase transition. It is noteworthy that a small change in the default correlation between credits leads to a drastic and sudden change in the mode's behavior. We also observe that at $\rho = 0$, the mode appears to stabilize. The mode's stationarity at $\rho = 0$ occurs because the loss distribution approaches that of the binomial distribution when the default correlations are close to zero. When $\rho = 0$, the loss distribution is exactly binomial with parameters $N$ and $p$.
In Figure \ref{fig:2figsprobmode} the mode probability is shown. We notice that it's high at extreme values of correlation and when the correlation is close to zero; for the other values of correlation it has a low value.

\begin{figure} 
\begin{minipage}[c]{0.4\linewidth}
\includegraphics[width=\linewidth]{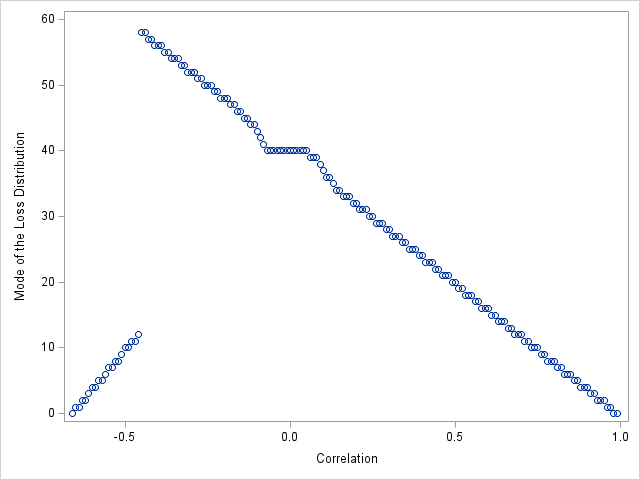}
\caption{Mode \emph{vs} Correlation}
\label{fig:2figsmode}
\end{minipage}
\hfill
\begin{minipage}[c]{0.4\linewidth}
\includegraphics[width=\linewidth]{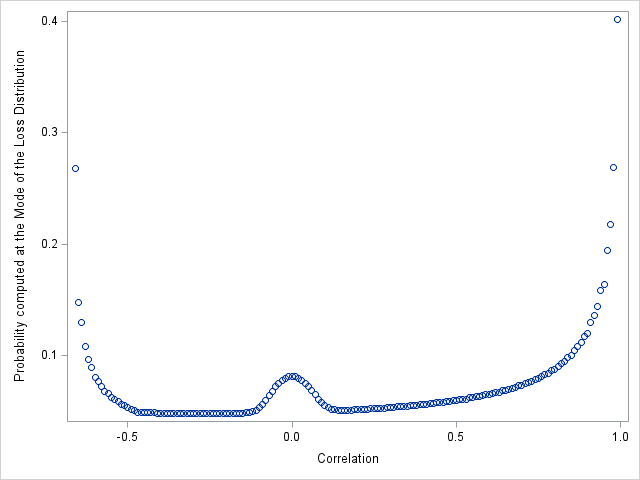}
\caption{Mode Probability \emph{vs} Correlation}
\label{fig:2figsprobmode}
\end{minipage}%
\end{figure}

\begin{comment}
 \begin{figure}
    \centering
    \includegraphics[width=0.45\linewidth]{image4.png}
    \caption{Enter Caption}
    \label{fig:2peaks}
\end{figure}
\begin{figure}
    \centering
    \includegraphics[width=0.45\linewidth]{image3.png}
    \caption{Enter Caption}
    \label{fig:mode}
\end{figure}   
\end{comment}

\begin{comment}
\begin{figure}
    \centering
    \includegraphics[width=0.45\linewidth]{image5.png}
    \caption{Enter Caption}
    \label{fig:enter-label}
\end{figure}    
\end{comment}

\section{Acknowledgments}
This work extends some of the results presented in the master's thesis \cite{tesiemonti}.

%
\bibliographystyle{spmpsci}
\bibliography{biblio.bib}

\end{document}